\documentclass[a4paper,UKenglish,autoref]{lipics-v2019}

\usepackage[ruled,vlined,linesnumbered,commentsnumbered]{algorithm2e}

\nolinenumbers

\title             {Approximation algorithms for connectivity augmentation problems}
\titlerunning{Approximation algorithms for connectivity augmentation problems}

\author{Zeev Nutov}{The Open University of Israel, nutov@openu.ac.il}{}{}{}

\authorrunning{Zeev Nutov}

\Copyright{Zeev Nutov}


\EventEditors{}
\EventNoEds{2}
\EventLongTitle{}
\EventShortTitle{}
\EventAcronym{}
\EventYear{2020}
\EventDate{}
\EventLocation{}
\EventLogo{}
\SeriesVolume{}
\ArticleNo{}

\begin{document}

\maketitle

\newcommand {\ignore} [1] {}

\newtheorem{fact}[lemma]{Fact}

\def\FF     {{\cal F}}
\def\II     {{\cal I}}

\def\eps  {\epsilon}
\def\al    {\alpha}

\def\empt {\emptyset}
\def\sem  {\setminus}
\def\subs  {\subseteq}


\def\f   {\frac}

\keywords{connectivity augmentation, approximation algorithm, element connectivity}

\begin{abstract}
In {\sc Connectivity Augmentation} problems we are given a graph $H=(V,E_H)$ and an edge set $E$ on $V$,
and seek a min-size edge set $J \subs E$ such that $H \cup J$ has larger edge/node connectivity than~$H$. 
In the {\sc Edge-Connectivity Augmentation} problem we need to increase the edge-connectivity by~$1$.
In the {\sc Block-Tree Augmentation} problem $H$ is connected and $H \cup S$ should be $2$-connected.
In {\sc Leaf-to-Leaf Connectivity Augmentation} problems every edge in $E$ connects minimal deficient sets.
For this version we give a simple combinatorial approximation algorithm with ratio $5/3$,
improving the $1.91$ approximation of \cite{BGA} (see also \cite{N-2c}), that applies for the general case.
We also show by a simple proof that if the {\sc Steiner Tree} problem admits approximation ratio $\al$ 
then the general version admits approximation ratio $1+\ln(4-x)+\eps$, 
where $x$ is the solution to the equation $1+\ln(4-x)=\al+(\al-1)x$.
For the currently best value of $\al=\ln 4+\eps$ \cite{BGRS} this gives ratio $1.942$.
This is slightly worse than the ratio $1.91$ of \cite{BGA}, 
but has the advantage of using {\sc Steiner Tree} approximation as a ``black box'', 
giving ratio $< 1.9$ if ratio $\al \leq 1.35$ can be achieved. 

In the {\sc Element Connectivity Augmentation} problem we are given a graph $G=(V,E)$, $S \subs V$, 
and connectivity requirements $r=\{r(u,v):u,v \in S\}$. The goal is to find a min-size set $J$ of new edges on $S$
(any edge is allowed and parallel edges are allowed) 
such that for all $u,v \in S$ the graph $G \cup J$ contains $r(u,v)$ $uv$-paths 
such that no two of them have an edge or a node in $V \sem S$ in common.
The problem is NP-hard even when $\displaystyle r_{\max} = \max_{u,v \in S} r(u,v)=2$. 
We obtain approximation ratio $3/2$, improving the previous ratio $7/4$ of \cite{N}.
For the case of degree bounds on $S$ we obtain the same ratio with just $+1$ degree violation, 
which is tight, since deciding whether there exists a feasible solution is NP-hard even when $r_{\max}=2$.
A similar result is shown for the more general problem 
of covering a skew-supermodular set function by a min-size set of edges.
\end{abstract}

\section{Introduction} \label{s:intro}

A graph is {\bf $k$-connected} if it contains $k$ internally disjoint paths between every pair of nodes;
if the paths are only required to be edge disjoint then the graph is {\bf $k$-edge-connected}.
In {\sc Connectivity Augmentation} problems we are given an ``initial'' graph $G_0=(V,E_0)$ and an edge set $E$ on $V$,
and seek a min-size edge set $J \subs E$ such that 
$G_0 \cup J=(V,E_0 \cup J)$ has larger edge/node connectivity than $G_0$. 
\begin{itemize}
\item 
In the {\sc Edge-Connectivity Augmentation} problem we seek to increase the edge connectivity by one,
so $G_0$ is $k$-edge-connected and $G_0 \cup J$ should be $(k+1)$-edge connected.
\item
In the {\sc $2$-Connectivity Augmentation} problem we seek to make a connected graph $2$-connected, 
so $G_0$ is connected and $G_0 \cup J$ should be $2$-connected.
\end{itemize}

A {\bf cactus} is a ``tree-of-cycles'', namely, 
a $2$-edge-connected graph in which every block is a cycle (equivalently - every edge belongs to exactly one simple cycle).
By \cite{DKL}, the {\sc Edge-Connectivity Augmentation} problem is equivalent to the following problem:

\begin{center} \fbox{\begin{minipage}{0.98\textwidth}
\underline{\sc Cactus Augmentation} \\
{\em Input:} \ \ A cactus $T=(V,E_T)$ and an edge set $E$ on $V$. \\
{\em Output:} A  min-size edge set $J \subs E$ such that $T \cup J$ is $3$-edge-connected.
\end{minipage}} \end{center}

It is also known (c.f. \cite{K-sur}) that the {\sc $2$-Connectivity Augmentation} problem is equivalent to the following problem:

\begin{center} \fbox{\begin{minipage}{0.98\textwidth}
\underline{\sc Block-Tree Augmentation} \\
{\em Input:} \ \ A tree $T=(V,E_T)$ and an edge set $E$ on $V$. \\
{\em Output:} A  min-size edge set $F \subs E$ such that $T \cup F$ is $2$-connected.
\end{minipage}} \end{center}

A more general problem than {\sc Cactus Augmentation} is as follows.
Two sets $A,B$ {\bf cross} if $A \cap B \neq \empt$ and $A \cup B \neq V$. 
A set family $\FF$ on a groundset $V$ is a {\bf crossing family} if $A \cap B,A \cup B \in \FF$ whenever $A,B \in \FF$ cross;
$\FF$ is a {\bf symmetric family} if $V \sem A \in \FF$ whenever $A \in \FF$.
The $2$-edge-cuts of a cactus form a symmetric crossing family, 
with the additional property that whenever $A,B \in \FF$ cross and $A \sem B,B \sem A$ are both non-empty, 
the set $(A \sem B) \cup (B \sem A)$ is not in $\FF$; 
such a symmetric crossing family is called {\bf proper} \cite{DN}. 
Dinitz, Karzanov, and Lomonosov \cite{DKL} showed that the family of
minimum edge cuts of a graph $G$ can be represented by $2$-edge cuts of a cactus. 
Furthermore, when the edge-connectivity of $G$ is odd, the min-cuts form a laminar family and thus can be represented by a tree.
Dinitz and Nutov \cite[Theorem 4.2]{DN} (see also \cite[Theorem 2.7]{N-Th})
extended this by showing that an arbitrary symmetric crossing family $\FF$ can be represented 
by $2$-edge cuts and specified $1$-node cuts of a cactus; 
when $\FF$ is a proper crossing family this reduces to the cactus representation of \cite{DKL}. 
We say that an edge $f$ {\bf covers} a set $A$ if $f$ has exactly one end in $A$.
The following problem combines the difficulties 
of the {\sc Cactus Augmentation} and the {\sc Block-Tree Augmentation} problems, see \cite{N-2c}.

\begin{center} \fbox{\begin{minipage}{0.98\textwidth}
\underline{\sc Crossing Family Augmentation} \\
{\em Input:} \ \ A graph $G=(V,E)$ and a symmetric crossing family $\FF$ on $V$. \\
{\em Output:} A  min-size edge set $J \subs E$ that covers $\FF$.
\end{minipage}} \end{center}

In this problem, the family $\FF$ may not be given explicitly, 
but we require that certain queries related to $\FF$ can be answered in polynomial time, see \cite{N-2c}. 
{\sc Block-Tree Augmentation} and {\sc Crossing Family Augmentation} admit ratio $2$ \cite{RW,FJW}, 
that applies also for the min-cost versions of the problems.

The inclusion minimal members of a set family $\FF$ are called {\bf leaves}.
In the {\sc Leaf-to-Leaf Crossing Family Augmentation} problem, every edge in $E$ connects two leaves of $\FF$.
In the {\sc Leaf-to-Leaf Block-Tree Augmentation} problem, every edge in $E$ connects two leaves of the input tree $T$.

\begin{theorem} \label{t:1}
The leaf-to leaf versions of {\sc Crossing Family Augmentation} and {\sc Block-Tree Augmentation} admit ratio $5/3$.
\end{theorem}

Better ratios are known for two special cases.
In the {\sc Tree Augmentation} problem the family $\FF$ is laminar, namely, any two sets in 
$\FF$ are disjoint or one contains the other; this problem can be also defined in connectivity terms - 
make a spanning tree $2$-edge-connected by adding a min-size edge set $J \subs E$. 
This problem was vastly studied; see \cite{A,GKZ,KN-TAP,FGKS,N-T} and the references therein 
for additional literature on the {\sc Tree Augmentation} problem.
In the {\sc Leaf-to-Leaf Tree Augmentation} problem, every edge in $E$ connects two leaves of the tree; 
this problem admits ratio $17/12$ \cite{MN}.
The {\sc Cycle Augmentation} problem is a particular case of the {\sc Cactus Augmentation} problem 
when the cactus is a cycle; in this case the leaves are the singleton nodes. 
The {\sc Cycle Augmentation} problem admits ratio $\f{3}{2}+\eps$ \cite{GGAS}; our 
algorithm from Theorem~\ref{t:1} uses some ideas from \cite{GGAS}.

Byrka, Grandoni, and  Ameli \cite{BGA} showed that {\sc Cactus Augmentation} 
admits ratio $2\ln 4- \f{967}{1120}+\eps<191$, breaching the natural $2$ approximation barrier. 
This was extended to {\sc Crossing Family Augmentation} and {\sc Block Tree Augmentation} in \cite{N-2c}. 

In the {\sc Steiner Tree} problem we are given a graph $G=(V,E)$ with edge costs and a set $R \subs V$ of terminals,
and seek a min-cost subtree of $G$ that spans $R$. We prove the following.

\begin{theorem} \label{t:2}
If {\sc Steiner Tree} admits ratio $\al$ then 
{\sc Crossing Family Augmentation} and {\sc Block-Tree Augmentation} admit ratio $1+\ln(4-x)+\eps$, 
where $x$ is the solution to the equation $1+\ln(4-x)=\al+(\al-1)x$.
\end{theorem}

Currently, $\al=\ln 4+\eps$ \cite{BGRS}; in this case we have 
ratio $1.942$ for the problems in the theorem. 
This is slightly worse than the ratio $1.91$ of \cite{BGA} (see also \cite{N-2c}), 
but our algorithm is very simple and has the advantage of using {\sc Steiner Tree} approximation as a ``black box''.
E.g., if ratio $\al=1.35$ can be achieved, then we immediately get ratio $1.895 < 1.9$. 

\medskip

We also consider the following problem:

\begin{center} \fbox{\begin{minipage}{0.975\textwidth} \noindent
\underline{{\sc Element Connectivity Augmentation}} \\
{\em Input:} \ An undirected graph $G=(V,E)$, a set $S \subs V$ of terminals, 
and connectivity requirements $\{r(u,v):u,v \in S\}$ on pairs of terminals.  \\
{\em Output:}   A minimum size set $J$ of new edges on $S$ (any edge is allowed and parallel edges are allowed) 
such that the graph $G \cup J$ contains $r(u,v)$ $uv$-paths such that no two of them have an edge or a node in $V \sem S$ in common.
\end{minipage}}\end{center}

A particular case when the graph $G$ is bipartite 
with sides $S$ and $V \sem S$ is known as the {\sc Hypergraph Edge-Connectivity Augmentation} problem;
here $S$ is the set of nodes of the hypergraph and $V \sem S$ is the set of the hyperedges. 
This problem is solvable in polynomial time for uniform requirements when $r(u,v)=k$ for all $u,v \in S$ \cite{BJ} 
(see also \cite{BF} and \cite{BK} for a simpler algorithm and proof), 
and when $r_{\max}=1$, where $r_{\max}$ is the maximum requirement.
See also \cite{F-book,FJ-survey,BK} for additional polynomially solvable cases.
The non-uniform version of the problem is NP-hard even when 
the initial graph $G$ is connected and $r_{\max}=2$ \cite{KCJ}.
The previous best approximation ratio for the general version was $7/4$, and $3/2$ when $r_{\max} =2$ \cite{N}.

In the degree bounded version of the problem 
we also have degree bounds $\{b(v):v \in S\}$ and require that 
$d_J(v) \leq b(v)$ for all $v \in S$, 
where $d_J(v)$ is the degree of $v$ w.r.t. $J$.  
We show that {\sc Element Connectivity Augmentation} admits ratio $3/2$, and 
that this ratio can be achieved also for the degree bounded version with only additive $+1$ degree violation;
a better degree approximation is unlikely, since deciding whether there exists a feasible solution 
is NP-hard even when $r_{\max}=2$ and $b_{\max} =1$ \cite{KCJ}.

\begin{theorem} \label{t:eca}
{\sc Element Connectivity Augmentation} admits approximation ratio $3/2$. 
Moreover, the degree bounded version admits a bicriteria approximation algorithm 
that computes a solution $J$ of size at most $3/2$ times the optimal 
such that $d_J(v) \leq b(v)+1$ for all $v \in S$.
\end{theorem}

The proof of this theorem is based on a generic algorithm 
for covering a skew-supermodular set function, as is explained in Section~\ref{s:element}.

\section{The leaf-to-leaf case (Theorem~\ref{t:1})}

We prove Theorem~\ref{t:1} for the {\sc Crossing Family Augmentation} problem, 
and later indicate the changes needed to adopt the proof for the {\sc Block-Tree Augmentation} problem.
We need some definition to describe the algorithm.
Let $\FF$ be a set family on $V$.
We say that $A \in \FF$ {\bf separates} $u,v \in V$ if $|A \cap \{u,v\}|=1$;
$u,v$ are {\bf $\FF$-separable} if such $A$ exists and $u,v$ are {\bf $\FF$-inseparable} otherwise.
Similarly, $A$ {\bf separates} edges $f,g$ if one of $f,g$ has both ends in $A$ and the other has no end in $A$;
$f,g$ are {\bf $\FF$-separable} if such $A \in \FF$ exists, and {\bf $\FF$-inseparable} otherwise.
The relation $\{(u,v) \in V \times V: u,v \mbox{ are } \FF\mbox{-inseparable}\}$ is an equivalence,
and we call its equivalence classes {\bf $\FF$-classes}.
W.l.o.g. we will assume that all $\FF$-classes are singletons and that no edge in $E$ has both ends in the same class; 
in particular, the leaves of $\FF$ are singletons, and we denote the leaf set of $\FF$ by $L$.
We will also often abbreviate the notation for singleton sets and write $v,e$ instead of $\{v\},\{e\}$.
Given $J \subs E$, the {\bf residual instance} $((V^J,E^J),\FF^J)$ is defined as follows.
\begin{itemize}
\item
The {\bf residual family} $\FF^J$ of $\FF$ w.r.t. $J$ consists of all members of $\FF$ that are uncovered by the edges in $J$.
It is known that $\FF^J$ is crossing (and symmetric) if $\FF$ is.
\item
$V^J$ is the set of $\FF^J$-classes (w.l.o.g, each of them can be shrunk into a single element).
\item 
$E^J$ is obtained from $E \sem J$ by removing all edges that have both ends in the same $\FF^J$-class.
\end{itemize}
In addition, given a set $R \subs V$ of terminals,
the {\bf residual set of terminals} $R^J$ is the set of $\FF^J$-classes that contain some member of $R$.
For illustration see Fig.~\ref{f:cactus}(a,b,c).

\begin{figure} 
\centering 
\includegraphics{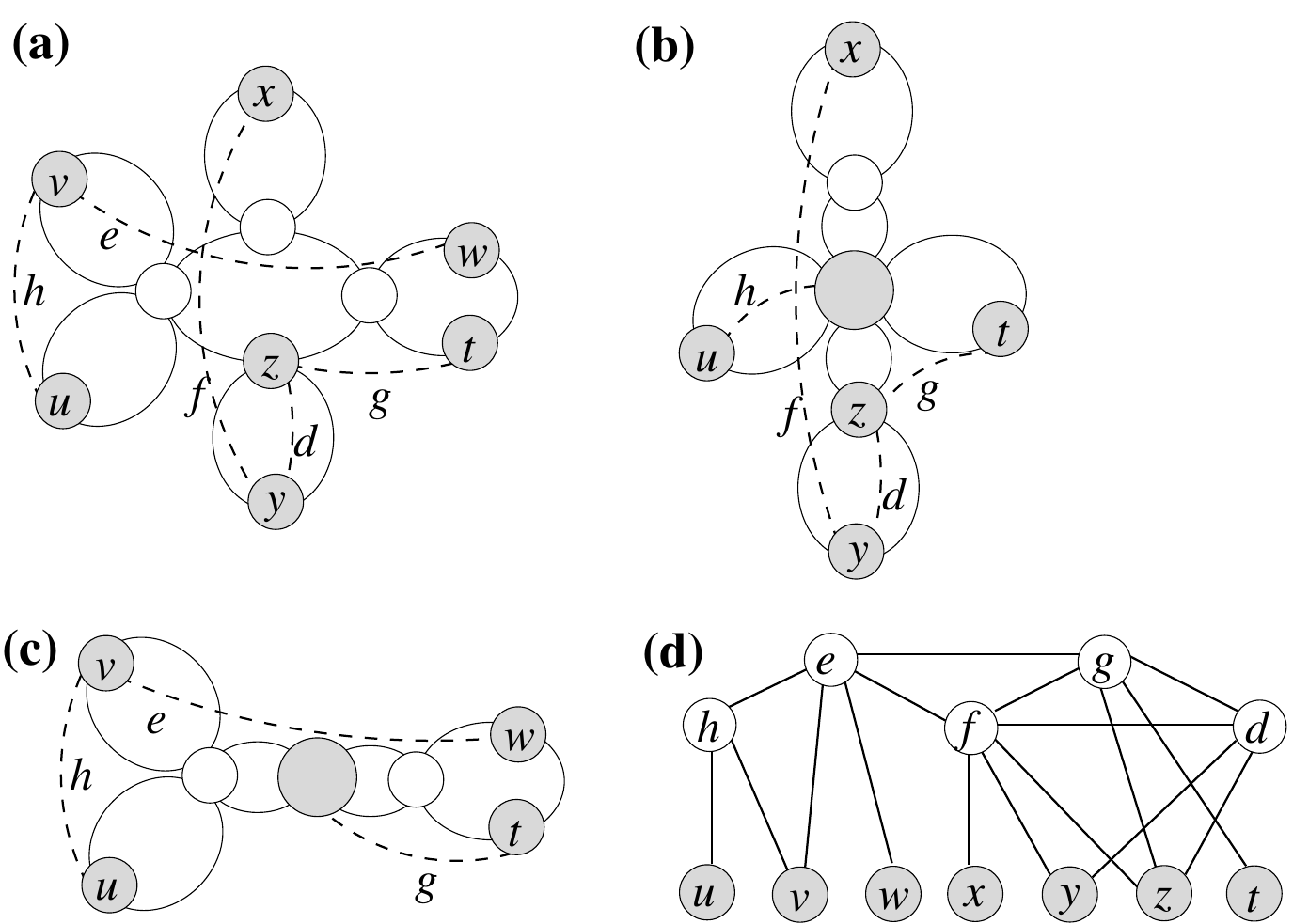}
\caption{Illustration of definitions for a {\sc Crossing Family Augmentation} instance where $\FF$ is represented by a cactus.
Here $A \in \FF$ if and only if $A$ is a connected component obtained 
by removing a pair of edges that belong to the same cycle of the cactus.
The edges in $E$ are shown by dashed arcs and the terminals in $R$ are shown by gray circles.
The cactus of the residual family w.r.t. to a single edge is obtained by ``squeezing'' the cycles 
along the path of cycles between the ends of the edge. 
(a)~The original instance. 
(b)~The residual instance w.r.t. $e$.
(c)~The residual instance w.r.t. $f$.
(d)~The $(R,E,\FF)$-incidence graph of the instance in (a).}
\label{f:cactus}
\end{figure}

\begin{figure} 
\centering 
\includegraphics{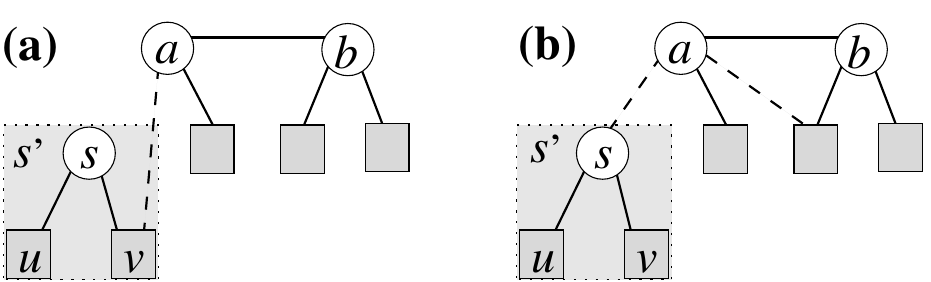}
\caption{Illustration to the proof of Lemma~\ref{l:lb}.}
\label{f:Rv}
\end{figure}

For any edge $e=uv$, there is an $\FF^e$-class that contains both $u$ and $v$; denote this class by $C(\FF,e)$. 
Given a set $R$ of terminals (a subset of $\FF$-classes), 
the {\bf $(R,E,\FF)$-incidence graph} $H=(U,E_H)$ has node set $U=E \cup R$ and edge set 
$$
E_H=\{ee':e,e' \in E \mbox{ are } \FF\mbox{-inseparable}\} \cup \{er:r \in R, e \in E, r \in C(\FF,e)\} \ .
$$

Let $R \subs V$ and let $H$ be the $(R,E,\FF)$-incidence graph.
Note that $R$ is an independent set in $H$.
It was shown in \cite{N-2c} that for $R=L$ being the set of leaves of $\FF$, 
an edge set $J \subs E$ is a feasible solution to {\sc Crossing Family Augmentation} if and only if 
the subgraph $H[J \cup R]$ of  $H$ induced by $J \cup R$ is connected.
The proof in \cite{N-2c} extends to any $R \subs V$ that contains~$L$. 
This implies that {\sc Crossing Family Augmentation} admits an approximation ratio preserving reduction 
to the following problem (see \cite{N-2c,BFGM} for more details).

\begin{center} \fbox{\begin{minipage}{0.98\textwidth}
\underline{\sc Subset Steiner Connected Dominating Set} ({\sc SS-CDS}) \\
{\em Input:} \ \  A  graph $H=(U,E_H)$ and a set $R \subs U$ of independent terminals.  \\
{\em Output:} A min-size node set $S \subs U \sem R$ such that $H[S]$ is connected and $S$ dominates $R$. 
\end{minipage}} \end{center}

Given a {\sc SS-CDS} instance and $s \in S=U \sem R$ let $R(s)=R_H(s)$ 
denote the set of neighbors of $s$ in $H$ that belong to $R$. 
Let ${\sf opt}$ be the optimal solution value of a problem instance at hand.
Before describing the algorithm, we will prove the following lemma.

\begin{lemma} \label{l:lb}
Let $\II=(H,R)$ be a {\sc SS-CDS} instance such that $|R(s)|=2$ for all $s \in S=U \sem R$. 
Then one of the following holds: 
\begin{itemize}
\item[{\em (i)}]
There are adjacent $a,b \in S$ with $R(a) \cap R(b)=\empt$. 
\item[{\em (ii)}]
${\sf opt} \geq |R|-1$.
\end{itemize}
\end{lemma}
\begin{proof}
Assume that (i) does not hold for $\II$; we will prove that then (ii) holds.
The proof is by induction on $|R|$. In the base case $|R|=2$ (ii) holds. 
Assume that the statement is true for $|R|-1 \geq 2$.
Let $T$ be an optimal solution tree and $S$ the set of non-terminals in $T$.
Root $T$ at some node and let $s \in S$ be a non-terminal farthest from the root. 
The children of $s$ are terminal leaves, and assume w.l.o.g. that $R(s)=\{u,v\}$ is the set of children of $s$;
if $s$ has just one child in $T$, then it has another terminal neighbor in $H$, that can be attached to $s$.

Consider the residual instance $\II'=(G'=(V',E'),R')$ and the tree $T'$ obtained by contracting $R(s)$ into the new terminal $s'$,
and deleting any $z \in U \sem (R+s)$ with $R(z)=R(s)$. 
Then $|R'|=|R|-1$, $|R'(z)|=2$ for all $z \in R'$, $T'$ is an optimal solution for $\II'$, and $S'=S-s$ is the set of non-terminals of $T'$. 

If (i) does not hold for the new instance $\II'$ then (ii) holds for $\II'$,  by the induction hypothesis.
Then $|S|=|S'|+1 \geq (|R'|-1)+1 = |R|-1$, and we get that (ii) holds for $\II$.
Assume henceforth that (i) holds for $\II'$. 
We obtain a contradiction by showing that then (i) holds for $\II$.
Let $a,b \in V' \sem R'$ be such that $R'(a) \cap R'(b)=\empt$, see Fig.~\ref{f:Rv}.
If $s' \notin R'(a) \cup R'(b)$ then clearly (ii) holds for $\II$.
Otherwise, if say $s' \in R'(a)$, then we have two cases.
If one of $u,v$, say $v$, is a neighbor of $a$ in $G$ (see Fig.~\ref{f:Rv}(a)) then  
$R(a) \cap R(b)=\empt$. Otherwise (see Fig.~\ref{f:Rv}(a)), $R(a) \cap R(s)=\empt$.
In both cases, we obtain a contradiction to the assumption that (i) does not hold for $\II$.
\end{proof}

We also need the following known lemma.

\begin{lemma} \label{l:minimal}
Any inclusion minimal cover $J$ of a set family $\FF$ is a forest.
\end{lemma}
\begin{proof}
Suppose to the contrary that $J$ contains a cycle $C$. 
Since $P=C \sem \{e\}$ is a $uv$-path, then for any $A$ covered by $e$, there is $e' \in P$ that covers $A$.
This implies that $J \sem \{e\}$ also covers $\FF$, contradicting the minimality of $J$. 
\end{proof}

The algorithm starts with a partial solution $J=\empt$  and has two phases.
Phase~1 consists of iterations. 
At the beginning of  each iteration, construct the $(E,R^J,\FF^J)$-incidence graph $H^J$, 
where initially $R$ is the set of leaves of $\FF$. Then, do one of the following:
\begin{enumerate}
\item
If $H^J$ has a node $e \in E$ with $|R^J(e)| \geq 3$, then add $e$ to $J$.  
\item
Else, if there are $e,f \in E$ with $R^J(e) \cap R^J(f)=\empt$, then add both $e,f$ to $J$.
\end{enumerate}
If none of the above two cases occurs, then we apply Phase~2,
in which we add to $J$ an inclusion minimal cover of $\FF^J$; 
note that all edges in $E^J$ have both endnodes in $R^J$.
A more formal description is given in Algorithm~\ref{alg:RS}.

We show that the algorithm achieves ratio $5/3$. Note that:
\begin{itemize}
\item
Adding an edge $e$ as in step~4 reduces the number of terminals by at least $2$. 
\item
Adding an edge pair $e,f$ as in step~5 reduces the number of terminals by at least $3$. 
\end{itemize}

\medskip 

\begin{algorithm}[H]
\caption{$(G=(V,E),\FF,R)$}  
\label{alg:RS}
$J \gets \empt$ \\
\Repeat{\em no edge $e$ or an edge pair $e,f$ as above exists}
{
let $H^J$ be the $(E^J,R^J,\FF^J)$-incidence graph \\
{\bf if} $H^J$ has a node $e \in E$ with $|R^J(e)| \geq 3$ then {\bf do} $J \gets J \cup \{e\}$ \\
{\bf else if} $H^J$ has node pair $e,f \in E$ with $R^J(e) \cap R^J(f)=\empt$ then {\bf do} $J \gets J \cup\{e,f\}$
}
find an inclusion minimal $\FF^J$-cover and add it to $J$ \\
\Return{$J$}
\end{algorithm}

\medskip

Hence the reduction in the number of terminals per added edge is at least $3/2$.
Let $\ell=|L|$ be the initial number of terminals. 
Let $\ell'=|R^J|$ be the number of terminals at the end of Phase~1 (steps 2-6 in Algorithm~\ref{alg:RS}).
Let $k$ be the number of edges added during Phase~1.
Then $\ell' \leq \ell-\f{3}{2}k$, hence $k \leq \f{2}{3}(\ell-\ell')$. 
The number of edges added at the second phase is at most $\ell'-1$, by Lemma~\ref{l:minimal};
note that every edge in $E^J$ has both ends in $R^J$ and that $|R^J|=\ell'$.
On the other hand, ${\sf opt} \geq \f{\ell}{2}$, and ${\sf opt} \geq \ell'-1$, by part (ii) of Lemma~\ref{l:lb}.
Summarizing, we have the following:
\begin{itemize}
\item
The solution size is at most $k+\ell'-1 \leq \f{2}{3}(\ell-\ell')+\ell'-1=(2\ell+\ell'-3)/3$.
\item
${\sf opt} \geq \ell/2$ and ${\sf opt} \geq \ell'-1$.
\end{itemize}
Thus the approximation ratio is bounded by 
$
\f{(2\ell+\ell'-3)/3}{\max\{\ell/2,\ell'-1\}} 
$.
If $\ell/2 \geq \ell'-1$ then
$$
\f{(2\ell+\ell'-3)/3}{\max\{\ell/2,\ell'-1\}} \leq \f{(2\ell+(\ell/2+1)-3)/3}{\ell/2} =\f{(5\ell/2-2)/3}{\ell/2} < \f{5}{3} \ .
$$
Else, $\ell/2 < \ell'-1$, and then
$$
\f{(2\ell+\ell'-3)/3}{\max\{\ell/2,\ell'-1\}} < \f{(4(\ell'-1)+\ell'-3)/3}{\ell'-1} =\f{(5\ell'-7)/3}{\ell'-1} < \f{5}{3} \ .
$$
In both cases the ratio is bounded by $5/3$.

\medskip

We now adjust the proof to the {\sc Block-Tree Augmentation} problem. 
Let $G=(V,E)$ be a connected graph.
A node $v$ is a {\bf cutnode} of $G$ if $G \sem \{v\}$ is disconnected; 
an inclusion maximal node subset whose induced subgraph is connected and has no cutnodes is a {\bf block} of $G$;
equivalently, $B$ is a block if it is the node set of an inclusion maximal $2$-connected subgraph or of a bridge. 
The {\bf block-tree} $T$ of $G$ has node set $C_G \cup {\cal B}_G$, 
where $C_G$ is the set of cutnodes of $G$ and ${\cal B}_G$ is the set of blocks of $G$;
$T$ has an edge for each pair of a block and a cutnode that belongs to that block.
It is known that every $v \in V \sem C_G$ belongs to a unique block, and that $T$ is a tree.
The {\bf block-tree mapping} $\psi:V \rightarrow C_G \cup {\cal B}_G$ of $G$ is defined by $\psi(v)=v$ is $v \in C_G$ 
and $\psi(v)$ is the block that contains $v$ if $v \in V \sem C_G$.

Given a {\sc Block-Tree Augmentation} instance $(T=(V,E_T),E)$ and $J \subs E$, 
the {\bf residual instance} $(T^J=(V^J,E_T^J),E^J)$ is defined as follows.
\begin{itemize}
\item
$T^J$ is the block tree of $T \cup J$.
\item 
$E^J=\{\psi(u)\psi(v):uv \in E \sem J, \psi(u) \neq \psi(v)\}$, where $\psi$ is the the block-tree mapping of $T \cup J$.  
\end{itemize}
For a set $R \subs V$ of terminals,
the {\bf residual set of terminals} is $R^J=\psi(R) = \cup_{r \in R} \psi(r)$. 
For an edge $e=uv$ let $T_e$ denote the unique $uv$-path in $T$.
We say that $e,f \in E$ are {\bf $T$-inseparable} if the paths $T_e,T_f$ have an edge in common.
The {\bf $(R,E,T)$-incidence graph} $H=(U,E_H)$ has node set $U=E \cup R$ and edge set 
$$
E_H=\{ef:e,f \in E \mbox{ are } T\mbox{-inseparable}\} \cup \{er:r \in R, e \in E, r \in T_e\} \ .
$$
It was shown in \cite{N-2c} that for $R=L$ being the set of leaves of $\FF$, 
an edge set $J \subs E$ is a feasible solution to {\sc Block-Tree Augmentation} if and only if 
the subgraph $H[J \cup R]$ of  $H$ induced by $J \cup R$ is connected.
The proof in \cite{N-2c} extends to any $R \subs V$ that contains~$L$. 
This implies that {\sc Crossing Family Augmentation} admits an approximation ratio preserving reduction 
to {\sc SS-CDS}, see \cite{N-2c} for details.
Lemma~\ref{l:minimal} also extends to this case, as it is known that 
an if $J$ is an inclusion minimal edge set whose addition makes a connected graph $2$-connected, then $J$ is a forest.

With these definitions and facts, the rest of the proof for the {\sc Block-Tree Augmentation} 
coincides with the proof given for {\sc Crossing Family Augmentation}, concluding the proof of Theorem~\ref{t:1}.

\section{The general case (Theorem~\ref{t:2})}

Recall that each of the problems {\sc Crossing Family Augmentation} and {\sc Block-Tree Augmentation} 
admits an approximation ratio preserving reduction 
to the {\sc SS-CDS} problem with $R=L$ being the set of terminals.
The {\sc SS-CDS} instances that arise from this reduction have the following property, see \cite{BGA,N-2c}: 
$$
(\ast) \ \ \ \mbox{The neighbors of every $r \in R$ induce a clique.}
$$

In fact, {SS-CDS} with property $(\ast)$ is equivalent to the 
{\sc Node Weighted Steiner Tree} problem with property $(\ast)$ with unit node weights for non-terminals (the terminals have weight zero).
Clearly, any {\sc SS-CDS} solution is a feasible {\sc Node Weighted Steiner Tree} solution;
for the other direction, note that if property $(\ast)$ holds, then the set of non-terminals 
in any feasible {\sc Node Weighted Steiner Tree} solution is a feasible {\sc SS-CDS} solution. 
The relation to the ordinary {\sc Steiner Tree} problem is given in following lemma. 

\begin{lemma} [\cite{BGA}] \label{l:al}
Let $S$ be a {\sc SS-CDS} solution 
and $T=(U,J)$ a {\sc Steiner Tree} solution on  instance $(G,R)$ with unit edge costs.
Then:
\begin{itemize}
\item[{\em (i)}]
If $(\ast)$ holds then $T$ can be converted into a {\sc SS-CDS} solution $S_J$ with $|S_J|=|J|-|R|+1$.  
\item[{\em (ii)}]
$S$ can be converted into a {\sc Steiner Tree} solution $T_S=(U_S,J_S)$ with $|J_S|=|S|+|R|-1$.
\end{itemize}
\end{lemma}
\begin{proof}
We prove (i). 
Any {\sc Steiner Tree} solution $T'=(U',J')$ can be converted into a solution 
$T=(U,J)$ such that $|J|=|J'|$ and $R$ is the leaf set of $T'$. For this, for each $r \in R$ that is not a leaf of $T'$,
among the edges incident to $r$ in $T'$, choose one and replace the other edges by a tree on the neighbors of $r$;
this is possible by $(\ast)$.
The non-leaf nodes of such $T$ form a a {\sc SS-CDS} as required.
For (ii), taking a tree on $S$ and for each $r \in R$ adding an edge from $r$ to $S$
gives a {\sc Steiner Tree} solution as required.
\end{proof}

Let $J^*$ be an optimal and $J$ an $\al$-approximate {\sc Steiner Tree} solutions.
Let $S_J,S^*$ be {\sc SS-CDS} solutions, where $S_J$ is derived from $J$ and $S^*$ is an optimal one. Then 
$$
|S_J| +R-1= |J| \leq \al|J^*| \leq \al |J_{S^*}| = \al(|S^*|-1+|R|) = \al|S^*|+\al(|R|-1) \ .
$$
This implies that if {\sc Steiner Tree} admits ratio $\al$ then {\sc SS-CDS} with property $(\ast)$ 
admits a polynomial time algorithm that computes a solution $S$ of size $|S| \leq \al{\sf opt}+(\al-1)|L|$ 
and achieves ratio $\al+(\al-1)\f{|L|}{\sf opt} = \al+(\al-1)x$, where $x=\f{|L|}{\sf opt}$, $0 < x \leq 2$.
We will prove the following.

\begin{theorem} \label{t:3}
{\sc Crossing Family Augmentation} and {\sc Block-Tree Augmentation} admit ratio $1+\ln \left(4-\f{|L|}{\sf opt} \right)+\eps$. 
\end{theorem}

From Lemma \ref{l:al} and Theorem~\ref{t:3} it follows that we can achieve ratio 
$$
\max\left\{\al+(\al-1)x,1+\ln \left(4-x\right)\right\}+\eps  \ \mbox{ where } \ x=\f{|L|}{\sf opt} \ .
$$
The worse case is when these two ratios are equal, which gives the Theorem~\ref{t:2} ratio.
In the case $\al=\ln 4+\eps$ \cite{BGRS}, we have 
$x \approx 1.4367$, so $L \approx 1.4367 {\sf opt}$ and  ${\sf opt} \approx 0.69 L$.
The ratio in this case is $1+\ln(4-x)+\eps < 1.942$.


\section{Proof of Theorem~\ref{t:3}}

A set function $f$ is 
{\bf increasing} if $f(A) \leq f(B)$ whenever $A \subs B$; $f$ is {\bf decreasing} if $-f$ is increasing, and $f$ 
is {\bf sub-additive} if $f(A \cup B) \leq f(A)+f(B)$ for any subsets $A,B$ of the ground-set.
Let us consider the following algorithmic problem:

\begin{center} \fbox{\begin{minipage}{0.96\textwidth} \noindent
\underline{{\sc Min-Covering}} \\
{\em Input:} \ 
Non-negative set functions $\nu,\tau$ on subsets of a ground-set $U$ such that 
$\nu$ is decreasing, $\tau$ is sub-additive, and $\tau(\empt)=0$. \\
{\em Output:}
$A \subs U$ such that $\nu(A)+\tau(A)$ is minimal.
\end{minipage}}\end{center} 

We call $\nu$ the {\bf potential} and $\tau$ the {\bf payment}.
The idea behind this interpretation and the subsequent greedy algorithm is as follows.
Given an optimization problem, the potential $\nu(A)$ is the (bound on the) value of some ``simple'' augmenting feasible solution for $A$. 
We start with an empty set solution, and iteratively try to decrease the potential by adding a set 
$B \subs U \sem A$ of minimum ``density'' -- the price paid for a unit of the potential. 
The algorithm terminates when the price $\geq 1$, since then we gain nothing from adding $B$ to $A$.
The ratio of such an algorithm is bounded by $1+\ln \f{\nu(\emptyset)}{\sf opt}$
(assuming that during each iteration a minimum density set  can be found in polynomial time).
So essentially, the greedy algorithm converts ratio $\al=\f{\nu(\emptyset)}{\sf opt}$ into ratio $1+\ln \al$.

Fix an optimal solution $A^*$.
Let $\nu^*=\nu(A^*)$, $\tau^*=\tau(A^*)$, so ${\sf opt}=\tau^*+\nu^*$. 
The quantity $\f{\tau(B)}{\nu(A)-\nu(A \cup B)}$ is called the {\bf density} of $B$ (w.r.t. $A$); this is the price 
paid by $B$ for a unit of potential.  
The {\sc Greedy Algorithm} (a.k.a. {\sc Relative Greedy Heuristic}) for the problem starts with $A=\empt$ and while $\nu(A)>\nu^*$ 
repeatedly adds to $A$ a non-empty augmenting set $B \subs U$ that satisfies the 
following condition, while such $B$ exists:

\medskip \noindent
{\bf Density Condition:}  \ 
$\displaystyle \f{\tau(B)}{\nu(A)-\nu(A \cup B)} \leq \min\left\{1,\f{\tau^*}{\nu(A)-\nu^*}\right\}$.
\medskip

Note that since $\nu$ is decreasing,
$\nu(A)-\nu(A \cup A^*) \geq \nu(A)-\nu(A^*)=\nu(A)-\nu^*$; hence if $\nu(A)>\nu^*$, then 
$\f{\tau(A^*)}{\nu(A)-\nu(A \cup A^*)} \leq \f{\tau^*}{\nu(A)-\nu^*}$
and there exists an augmenting set $B$ that satisfies the condition  
$\f{\tau(B)}{\nu(A)-\nu(A \cup B)} \leq \f{\tau^*}{\nu(A)-\nu^*}$, e.g., $B=A^*$.
Thus if $B^*$ is a minimum density set and $\f{\tau(B^*)}{\nu(A)-\nu(A \cup B^*)} \leq 1$,
then $B^*$ satisfies the Density Condition; 
otherwise, the density of $B^*$ is larger than $1$ so no set can satisfy the Density Condition. 
The following statement is known, c.f. an explicit proof in \cite{NS}.

\begin{theorem} \label{t:ga}
The {\sc Greedy Algorithm} achieves approximation ratio
$
1+\f{\tau^*}{\sf opt} \ln \f{\nu(\empt)-\nu^*}{\tau^*}
$.
\end{theorem}

This applies also in the case when we can only compute a $\rho$-approximate minimum density augmenting set,
while invoking an additional factor $\rho$ in the ratio.

To use the framework of Theorem~\ref{t:ga} we need to define $\tau$ and $\nu$.
Let $J \subs E$ be an edge set. The payment $\tau(J)=|J|$ is just the size of $J$. 
The potential of $J$ is defined by $\nu(J)=|R^J|-1$, where $R$ is a set of terminals such that $L \subs R \subs V$,
defined in the following lemma.
For an edge set $F$ let $F_{LL}$ be the set of edges in $F$ with both ends in $L$, 
and $F_L$ the set of edges in $F$ that have exactly one end in $L$.

\begin{lemma} \label{l:R}
Let $F$ be an optimal solution to {\sc Crossing Family Augmentation} instance and $c$ be a cost function on $E$ 
defined by $c(e)=0$ if $e \in E_{LL}$, $c(e)=1$ if $e \in E_L$, and $c(e)=2$ otherwise.
Let $J$ be a $2$-approximate $c$-costs solution and let $R$ be the set of ends of the edges in $J$. 
Then $|R| \leq c(J)+L \leq 4|F|-|L|=4{\sf opt}-|L|$.
\end{lemma}
\begin{proof}
Clearly, $|R| \leq c(J)+|L|$. We show that $c(J) \leq 4|F|-2|L|$.
Let $F'$ be the set of edges in $F$ that have no end in $L$.
Since $|F'| = |F| - |F_L|-|F_{LL}|$ and $2|F_{LL}|+|F_L| \geq L$ 
$$
c(F) =|F_L|+2|F'| =|F_L|+2(|F|-|F_L|-|F_{LL}|)=2|F|-(|F_L|+2|F_{LL}|) \leq 2|F|-|L| \ .
$$
Since $c(J) \leq 2c(F)$, the lemma follows.
\end{proof}

It is easy to see that $\nu$ is decreasing and $\tau$ is subadditive. 
The next lemma shows that the obtained {\sc Min-covering} instance is equivalent to the 
{\sc Crossing Family Augmentation} instance, and that we may assume that $\tau^*={\sf opt}$ and $\nu^*=0$.

\begin{lemma} \label{l:fe}
If $J$ is a feasible solution to {\sc Crossing Family Augmentation} then $\nu(J)=0$. 
If $J$ is a feasible {\sc Min-Covering} solution then one can construct in polynomial time 
a feasible {\sc Crossing Family Augmentation} solution of size $\leq \tau(J)+\nu(J)$. 
In particular, both problems have the same optimal value,
and {\sc Min-Covering} has an optimal solution $J^*$ such that $\nu(J^*)=0$ and $\tau(J^*)={\sf opt}$.
\end{lemma}
\begin{proof}
If $J$ is a feasible {\sc Crossing Family Augmentation} solution then $|R^J|=1$ and thus $\nu(J)=0$.
Let $I$ be a {\sc Min-Covering} solution such that every edge in $I$ has both ends in $R$; e.g., $I$ can be as in Lemma~\ref{l:R}. 
Then $I^J$ is a feasible solution to the residual problem w.r.t. $J$ and every edge in $I^J$ has both ends in $R^J$. 
Let $I' \subs I^J$ be an inclusion minimal edge set such that $J \cup I'$ is a feasible solution.
By Lemma~\ref{l:minimal}, $I'$ is a forest, hence $|I| \leq |R^J|-1$. 
Consequently, $J \cup I'$ is a feasible solution of size at most $|J|+|I'| \leq |J|+|R^J|-1=\tau(J)+\nu(J)$.
\end{proof}

Recall also that $\nu(\empt) \leq 4{\sf opt}-|L|$, by Lemma~\ref{l:R}.
We will show how to find for any $\eps>0$, a $(1+\eps)$-approximate best density set in polynomial time. 
It follows therefore that we can apply the greedy algorithm to produce a solution of value $1+\eps$ times of
$$
1+\f{\tau^*}{\sf opt} \ln \f{\nu(\empt)-\nu^*}{\tau^*}=1+\ln \f{4 {\sf opt} -|L|}{\sf opt} = 1+\ln \left(4-\f{|L|}{\sf opt} \right) \ .
$$


In what follows note that if $a_1, \ldots,a_q$ and $b_1, \ldots b_q$ are positive reals,
then by an averaging argument there exists an index $1 \leq i \leq q$ such that $a_i/b_i \leq \sum_{j=1}^q a_j/\sum_{i=1}^q b_j$.

Given a {\sc Crossing Family Augmentation} instance, a set $R \supseteq L$ of terminals, and $F \subs E$,
consider the corresponding {\sc SS-CDS} instance $(H=(U,E_H),R)$ and the set of non-terminals $Q$ that corresponds to $F$.
The density of $F$ is $\f{|F|}{|R|-|R^F|}$, and in the {\sc SS-CDS} instance this is computed 
by taking a maximal forest in the graph induced by $Q$ and the terminals that have a neighbor in $Q$;
then the density is $|Q|$ over the number of trees in this forest. 
So in what follows we may speak of a density of a subforest of $H$. 
Let $T_i=(S_i \cup R_i,E_i)$, $i=1, \ldots,q$,  be the connected components of such a forest,
($R_i$ is the set of terminals in $T_i$) and let $s_i=|S_i|$ and $r_i=|R_i|$, where $r_i \geq 2$. 
The density of the forest is $\sum_{i=1}^q s_i/ \sum_{i=1}^q (r_i-1)$ 
while the density of each $T_i$ is $s_i/(r_i-1)$. 
By an averaging argument, some $T_i$ has density not larger than that of the forest.
Consequently, we may assume that the minimum density is attained for a tree, say $T$.  

Let $T=(S \cup R,E)$ be a tree with leaf set $R$. The density of $T$ is $\f{s}{r-1}$, where $r=|R|$ is the number of terminals ($R$-nodes) 
and $s$ is the number of non-terminals ($S$-nodes) in $T$. 
The usual approach is to show that for any $k$ there exists a subtree $T'$ of $T$ with 
$k$ terminals (or $k$ non-terminals) 
such that the density of $T'$ is at most $1+f(k)$ times the density of $T$, where $\lim_{k \rightarrow \infty} f(k)=0$.
The decomposition lemma that we prove is not a standard one. 
The difficulty can be demonstrated by the following examples.
Consider the case when $T$ is a star with $n$ leaves. 
Then the density of $T$ is $1/(n-1)$, while a subtree with $k$ leaves has density $1/(k-1)$. 
If $T$ is a path with $n$ non-terminals, then the density of $T$ is $n$, 
while a subtree with $k<n$ non-terminals has density $k/0=\infty$. 
In both cases, the density of the subtree may be arbitrarily larger than that of $T$.
To overcome this difficulty, we will decompose $T$ w.r.t. a certain subset $P$ of the non-terminals.

Let $P \subs S$. Let $s=|S|$, $r=|R|$, and $p=|P|$.
For a subtree $T'$ of $T$ let $S(T')$, $R(T')$, and $P(T')$ denote 
the set of $S$-nodes, $R$-nodes, and $P$-nodes in $T'$, respectively. 
We prove the following.

\begin{lemma} \label{l:td}
Let $k \geq 2$. If $p \geq 3k+1$ then there exists subtrees $T_1, \ldots, T_q$ of $T$ such that the following holds. 
\begin{itemize}
\item
$\sum_{i=1}^q |S(T_i)| \leq s+q$. 
\item
Every $R$-node belongs to exactly one subtree, hence  $\sum_{i=1}^q |R(T_i)| =r$.
\item
$|P(T_i)| \in [k,3k]$ for all $i$ and $q \leq \f{p}{k-1}$.
\end{itemize}
\end{lemma}
\begin{proof}
Root $T$ at some node in $S$. For any $v \in S$ chosen as a ``local root'', 
the subtree $T^v$ rooted at $v$ is a subtree of $T$ that consist of $v$ and its descendants. 
Let $T^v$ be an inclusion minimal rooted subtree of $T$ such that $|P(T^v)| \geq k+1$.
Note that $v \in P$. Let $B_1, \ldots,B_m$ be the branches hanging on $v$ and let $p_j=|P(B_j)|$.
By the definition of $T_v$, each $p_j$ is in the range $[0,k]$ and $\sum_{j=1}^m p_j \geq k$. 
We claim that $\{p_1, \ldots,p_m\}$ can be partitioned such that the sum of each part plus $1$ is in the range $[k,3k]$.
To see this, apply a greedy algorithm for {\sc Multi-Bin Packing} with bins of capacity $2k$;
at the end there is at most one bin with sum $\leq k-1$ (as two such bins can be joined),
and joining this bin to any other bin gives a partition as required.
Now we remove $T^v$ and the $S$-nodes on the path from $v$ to its closest terminal ancestor, 
and apply the same procedure on the remaining tree.
If the last rooted subtree $T^v$ considered has $|P(T^v)| \leq k-1$, then this tree can be joined to 
a subtree $T_i$ with $|P(T_i)| \leq 2k$ derived in previous iteration.  
Finally, $q \leq \f{p+q}{k}$ by the construction and since $|P(T_i)| \geq k$ for all $i$;
this implies $q \leq \f{p}{k-1}$.
\end{proof}

Now we let $P=P_1 \cup P_2$, where 
$P_1$ is the set of nodes that have degree at least $3$ in $T$ and 
$P_2$ is the of nodes that have a terminal neighbor in $T$.
Note that $|P_1| \leq r$ and $|P_2| \leq r$. Hence $p \leq 2r$, and clearly $p \leq s$.
By an averaging argument and Lemma~\ref{l:td}, the density of some $T_i$ is bounded by
$s_i/(r_i-1) \leq \sum_{j=1}^q s_j/\sum_{j=1}^q (r_j-1) \leq (s+q)/(r-q)$. 
Thus for $k \geq 3$ we get
$$
\f{s_i}{r_i-1} \cdot \f{r-1}{s} \leq \f{s+q}{r-q}  \cdot \f{r}{s} \leq  \f{s+p/(k-1)}{r-p/(k-1)}  \cdot \f{r}{s} = 
\f{1+1/(k-1)}{1-2/(k-1)} =\f{k}{k-3} =1+\f{3}{k-3} \ .
$$

This implies that we can find a $(1+\eps)$-approximate min-density tree by searching
over all trees $T'$ with $|P(T')| \in [k,3k]$, where given $\eps>0$ we let $k=\lceil 3/\eps \rceil+3$.
Specifically, for every $P' \subs S$ with $|P'| \in [k,3k]$, we find an MST $T'$ in the metric completion 
of the current incidence graph, and then add to $T'$ all the terminals that have a neighbor in $P'$.
Among all subtrees we choose one of minimum density.  
The time complexity is $n^{3k}$ which is polynomial for any fixed $\eps >0$.

The process of adjusting the proof to the {\sc Block-Tree Augmentation} is identical to the one in the proof of Theorem~\ref{t:1}.   
This concludes the proof of Theorem~\ref{t:3}, and thus also the proof of Theorem~\ref{t:2} is complete.

\section{Covering skew-supermodular functions (Theorem~\ref{t:eca})} \label{s:element}
 
Let  $p:2^S \rightarrow \mathbb{Z}$ be a set function and $J$ an edge set on a finite groundset $S$. 
We say that {\bf $J$ covers $p$} if $d_J(A) \geq p(A)$ for all $A \subs S$,
where $d_J(A)$ denote the set of edges with exactly one end in $A$.
$p$ is {\bf symmetric} if $p(A)=p(S \sem A)$ for all $A \subs S$, and $p$
is {\bf skew-supermodular} (a.k.a. {\bf weakly supermodular}) 
if for all $A,B \subs S$ at least one of the following two inequalities holds: 
$$ 
p(A)+p(B) \leq p(A \cap B)+p(A \cup B)  \ \ \ \ \ 
p(A)+p(B) \leq p(A \sem B)+p(B \sem A)
$$
{\sc Element Connectivity Augmentation} can be reduced to the following problem, 
with skew-supermodular set function $p$, c.f. \cite{FJ-survey,N,BK}.

\begin{center} \fbox{\begin{minipage}{0.975\textwidth} \noindent
\underline{{\sc Set Function Edge Cover}} \\
{\em Input:} \ \ A set function $p$ on a ground-set $S$.  \\
{\em Output:}  A minimum size set $J$ of edges that covers $p$.
\end{minipage}}\end{center}

\noindent
In this problem, the function $p$ may not be given explicitly, and a polynomial time implementation of algorithms 
requires that some queries related to $p$ can be answered in polynomial time.
But the problem is also NP-hard for skew-supermodular $p$ even if $p$ given explicitly, specifically 
when $p_{\max}=1$ and $|A|=3$ for every set $A$ with $p(A)=1$ \cite{N}.

In the degree bounded version of this problem we are also given degree bounds $\{b(v):v \in S\}$ 
and require that $d_J(v) \leq b(v)$ for all $v \in S$.

\begin{definition}
A function $g:S \rightarrow \mathbb{Z}_+$ is a {\bf $p$-transversal} if $g(A) \geq p(A)$ for all $A \subs S$.
Let $T_g=\{v \in S:g(v) \geq 1\}$ denote the support of $g$.
We say that $g$ is a {\bf minimal $p$-transversal} if for any $v \in T_g$ reducing $g(v)$ by $1$
results in a function that is not a $p$-transversal.
\end{definition}

The following was proved by Bencz\'{ur} and Frank in \cite{BF}, see also \cite[Lemmas 1.1 and 3.2]{N}.

\begin{lemma} \label{l:Tg}
Let $g$ be a transversal of a skew-supermodular set function $p$. Then:
\begin{itemize}
\item
$g(S)=\max\{\sum_{A \in \FF}p(A):\FF \mbox{ is a subpartition of } S\}$ if $g$ is minimal.
\item
There exists an optimal $p$-cover $J$ such that every $e \in J$ has both ends in $T_g$. 
\end{itemize}
\end{lemma}

Let $\tau(p)$ denote the size of a minimal $p$-cover.
As $g=\{d_J(v):v \in S\}$ is a $p$-transversal for any $p$-cover $J$, 
$\tau(p) \geq g(S)/2$ for any minimal $p$-transversal $g$. 
Thus a natural approach to compute a small $p$-cover is: repeatedly choose an edge $uv$
with $u,v \in T_g$, such that updating $p$ and reducing $g(u)$ and $g(v)$ by $1$, 
keeps $g$ being a $p$-transversal. This approach works for many interesting special cases, c.f. \cite{BK},
but in general such an edge $uv$ may not exist.
Formally, given $u,v \in T_g$ define $p^{uv}$ and $g^{uv}$ by:
$$p^{uv}(A)=\max\{p(A)-1,0\} \mbox{ if } |A \cap \{u,v\}|=1 \mbox{ and } p^{uv}(A)=p(A) \mbox{ otherwise;}$$
$$g^{uv}(w)=g(w)-1 \mbox{ if } w=u \mbox{ or if } w=v \mbox{ and } g^{uv}(w)=g(w) \mbox{ otherwise.}$$
It is easy to see that if $p$ is (symmetric) skew-supermodular, so is $p^{uv}$. 
However, $g^{uv}$ may not be a $p^{uv}$-transversal if $g$ is. 
We say that a pair $u,v \in T_g$ is {\bf $(p,g)$-legal} if $g^{uv}$ is a $p^{uv}$-transversal; 
then replacing $p,g$ by $p^{uv},g^{uv}$ is the {\bf splitting-off} operation at $u,v$. 
Intuitively, splitting-off is an attempt to add the edge $uv$ to a partial solution,
and to consider the residual problem of covering $p^{uv}$ with the residual lower bound 
$\lceil g^{uv}(S)/2 \rceil=\lceil g(S)/2 \rceil-1$.
We need the following result due to \cite{N}, see also \cite{BK} for a short and elegant proof.

\begin{lemma}[\cite{N}] \label{l:legal}
Let $p$ be symmetric skew-supermodular and $g$ a $p$-transversal.
If $p_{\max} \geq 2$ then there exists a $(p,g)$-legal pair.
\end{lemma}

Lemma~\ref{l:legal} implies that if no $(p,g)$-legal pair exists, 
then any inclusion minimal solution on $T_g$ is a forest, and that any tree on $T_g$ is a feasible solution.
In \cite{N,BK} was considered a simple greedy algorithm which repeatedly splits-off legal pairs 
as long as such exist, and then adds to the partial solution a tree (or any inclusion minimal solution) on $T_g$.

\medskip \medskip

\begin{algorithm}[H]
\caption{{\sc Greedy$(p,g)$} \newline ($p$ is symmetric skew-supermodular, $g$ is a $p$-transversal)}
\label{alg:GC}
$M \gets \empt$ \\
\While{\em there exists a $(p,g)$-legal pair $u,v$}
{$g \gets g^{uv}$, $p \gets p^{uv}$, $M \gets M+uv$}
let $F$ be a tree on $T'=\{v \in S:g(v)=1\}$ \\
\Return{$M \cup F$}  
\end{algorithm}

\medskip \medskip

In the degree bounded version we let $g=\{b(v):v \in S\}$.
If this $g$ is not a $p$-transversal, then the problem has no feasible solution.
To get a degree violation $+1$, at step~4 of the algorithm we choose $F$ to be a path on $T'$.

In \cite{N} is was shown that for skew-supermodular $p$ this algorithm achieves ratio $7/4$,
by characterizing those pairs $p,g$ for which no $(p,g)$-legal pair exists and 
deriving a lower bound on $\tau(p)$. 
We establish a better lower bound than that of \cite{N} and prove the following.

\begin{theorem} \label{t:sfec}
Algorithm {\sc Greedy} achieves approximation ratio $3/2$.
Moreover, if  $F$ is chosen to be a path at step~4, then $d_J(v) \leq g(v)+1$ for all $v \in S$. 
\end{theorem}

Theorem~\ref{t:sfec} second statement is obvious, so in the rest of this section we prove the first statement.
The following can be deduced from Lemma~\ref{l:legal}, see \cite{N}.

\begin{corollary} [\cite{N}] \label{c:legal}
Let $p'$ be symmetric skew-supermodular and 
$g'$ a minimal $p'$-transversal with non-empty support $T'=\{v \in V:g'(v) \geq 1\}$,
and suppose that no $(p',g')$-legal pair exists. 
Then $p'_{\max}=g'_{\max}=1$, $|T'| \geq 3$, and for every $A' \subs T'$ with $|A'| \in \{1,2\}$ there is 
$A \subs V$ with $p'(A)=1$ such that $A \cap T'=A'$.
Furthermore, $\tau(p') \geq \f{2}{3}|T'|$.
\end{corollary}

We now describe the analysis of the $7/4$-approximation \cite{N}.
Let $k$ be the number of edges accumulated in $M$ during the while-loop.
Let $t=g(S)$ be the initial value of the $p$-transversal $g$ and let $t'=|T'|$ be the 
the transversal value at the beginning of step~4.
Note that $t-t'=2k$ and that $|F|=|T'|-1 = t-2k-1$.
Consequently, $|M \cup F| \leq k+(t-2k-1) \leq t-k$.
On the other hand we have the lower bounds $\tau(p) \geq t/2$ and $\tau(p) \geq \f{2}{3}|T'|=\f{2}{3}(t-2k)$.
Thus the approximation ratio is bounded by
$\f{t-k}{\max\{t/2,2(t-2k)/3\}} \leq 7/4$, with $k=t/8$ being the worse case. 

One can observe that if $k \geq t/4$ then $t/2 \geq \f{2}{3}(t-2k)$, and thus in this case 
the ratio is bounded by $\f{t-k}{t/2} \leq \f{3/4}{1/2}=3/2$.
To get ratio $3/2$ for the range $k \leq t/4$ we give a better lower bound on $\tau(p)$.
For this, we need the following lemma.

\begin{lemma} \label{l:T}
Let $g$ be a minimal transversal of a skew-supermodular symmetric set function $p$.
Then there exists an optimal $p$-cover $J$ such that 
$d_J(v) \geq g(v)$ if $v \in T_g$ and $d_J(v)=0$ otherwise.
\end{lemma}
\begin{proof}
By induction on $\tau(p)$. The base case $\tau(p)=1$ is trivial. For $\tau(p) \geq 2$,
let $J$ be an optimal $p$-cover such that every $e \in J$ has both ends in $T_g$; 
such exists by Lemma~\ref{l:Tg}.
Choose some $e=uv \in J$ and let $p'=p^{uv}$. 
Let $g^u$ be obtained from $g$ by decreasing $g(u)$ by $1$, and similarly $g^v$ is defined.
Then one of $\{g,g^u,g^v,g^{uv}\}$ is a minimal $p'$-transversal; denote it by $g'$.
By the induction hypothesis there exists a $p'$-cover $J'$ such that: 
$d_{J'}(w) \geq g'(w)$ if $w \in T_{g'}$ and $d_{J'}(w)=0$ otherwise.
It is easy to see that $J=J' \cup \{e\}$ has the required property.
\qed
\end{proof}

\begin{lemma} \label{l:new}
$\tau(p) \geq \f{2}{3}(t-k)$. 
\end{lemma}
\begin{proof}
Let $J$ be a $p$-cover as in Lemma~\ref{l:T}.
Let $X$ be the set of edges in $J$ with both ends in $T_g \sem T'$
and $Y$ the set of edges in $J$ with exactly one end in $T_g \sem T'$; let $x=|X|$ and $y=|Y|$.
Since $d_J(v) \geq g(v)$ for all $v \in T$, 
$2x+y \geq t-t'=2k$, hence $y \geq 2k-2x$.
Let $Q$ be the set of nodes in $T'$ that are uncovered by edges in $X \cup Y$
and let $z$ be the number of edges in $J$ with both ends in $Q$.
Note that $|Q| \geq t'-y=t-2k-y$. 
By Corollary~\ref{c:legal}, for every $A' \subs Q$ with $|A'| \in \{1,2\}$ there is 
$A \subs V$ with $p(A)>0$ such that $A \cap T'=A'$.
This implies that $z \geq \f{2}{3}|Q| \geq \f{2}{3}(t-2k-y)$.
Consequently, since $|J| \geq x+y+z$ and $y \geq 2(k-x)$
$$
|J| \geq x+y+\f{2}{3}(t-2k-y)=x+\f{1}{3}y+\f{2}{3}(t-2k) \geq	x+\f{2}{3}(k-x)+\f{2}{3}(t-2k)=\f{1}{3}x+\f{2}{3}(t-k) \ .
$$
Since $x \geq 0$ we get $|J| \geq \f{2}{3}(t-k)$.
\qed
\end{proof}

From Lemma~\ref{l:new} it follows that the approximation ratio of the algorithm is bounded by
$\f{t-k}{\max\{t/2,2(t-k)/3\}} \leq 3/2$, with $k=t/4$ being the worse case. 

\medskip

This concludes the proof of Theorem~\ref{t:sfec}, and thus also the proof of Theorem~\ref{t:eca} is complete.


\begin{thebibliography}{10}

\bibitem{A}
D.~Adjiashvili.
\newblock Beating approximation factor two for weighted tree augmentation with
  bounded costs.
\newblock In {\em SODA}, pages 2384--2399, 2017.

\bibitem{BJ}
J.~Bang-Jensen and B.~Jackson.
\newblock Augmenting hypergraphs by edges of size two.
\newblock {\em Math. Programming}, 84:457--481, 1999.

\bibitem{BFGM}
M.~Basavaraju, F.~V. Fomin, P.~A. Golovach, P.~Misra, M.~S. Ramanujan, and
  S.~Saurabh.
\newblock Parameterized algorithms to preserve connectivity.
\newblock In {\em ICALP, {\em Part I}}, pages 800--811, 2014.

\bibitem{BF}
A.~Bencz\'{u}r and A.~Frank.
\newblock Covering symmetric supermodular functions by graphs.
\newblock {\em Math. Programming}, 84:483--503, 1999.

\bibitem{BK}
A.~Bern\'{a}th and T.~Kir\'{a}ly.
\newblock A unifying approach to splitting-off.
\newblock {\em Combinatorica}, 32(4):373--401, 2012.

\bibitem{BGA}
J.~Byrka, F.~Grandoni, and A.~J. Ameli.
\newblock Breaching the 2-approximation barrier for connectivity augmentation:
  a reduction to steiner tree.
\newblock In {\em STOC}, pages 815--825, 2020.
\newblock For the full version see.
\newblock URL: \url{https://arxiv.org/abs/1911.02259}.

\bibitem{BGRS}
J.~Byrka, F.~Grandoni, T.~Rothvo\ss, and L.~Sanità.
\newblock Steiner tree approximation via iterative randomized rounding.
\newblock {\em J. ACM}, 60(1):6:1--6:33, 2013.

\bibitem{DKL}
E.~A. Dinic, A.~V. Karzanov, and M.~V. Lomonosov.
\newblock On the structure of a family of minimal weighted cuts in a graph.
\newblock {\em Studies in Discrete Optimization}, page 290–306, 1976.

\bibitem{DN}
Y.~Dinitz and Z.~Nutov.
\newblock A 2-level cactus model for the system of minimum and minimum+1
  edge-cuts in a graph and its incremental maintenance.
\newblock In {\em STOC}, pages 509--518, 1995.

\bibitem{FGKS}
S.~Fiorini, M.~Gro\ss, J.~K\"{o}nemann, and L.~Sanit\'{a}.
\newblock A $\frac{3}{2}$-approximation algorithm for tree augmentation via
  chv\'{a}tal-gomory cuts.
\newblock In {\em SODA}, pages 817--831, 82018.

\bibitem{FJW}
L.~Fleischer, K.~Jain, and D.~P. Williamson.
\newblock Iterative rounding 2-approximation algorithms for minimum-cost vertex
  connectivity problems.
\newblock {\em J. Comput. Syst. Sci.}, 72(5):838--867, 2006.

\bibitem{F-book}
A.~Frank.
\newblock {\em Connections in Combinatorial Optimization}.
\newblock Oxford University Press, 2011.

\bibitem{FJ-survey}
A.~Frank and T.~Jord\'{a}n.
\newblock Graph connectivity augmentation.
\newblock In K.~Thulasiraman, S.~Arumugam, A.~Brandstadt, and T.~Nishizeki,
  editors, {\em Handbook of Graph Theory, Combinatorial Optimization, and
  Algorithms}, chapter~14, pages 313--346. CRC Press, 2015.

\bibitem{GGAS}
W.~G\'{a}lvez, F.~Grandoni, A.~J. Ameli, and K.~Sornat.
\newblock On the cycle augmentation problem: Hardness and approximation
  algorithms.
\newblock In {\em WAOA}, pages 138--153, 2019.

\bibitem{GKZ}
F.~Grandoni, C.~Kalaitzis, and R.~Zenklusen.
\newblock Improved approximation for tree augmentation: saving by rewiring.
\newblock In {\em STOC}, pages 632--645, 2018.

\bibitem{K-sur}
S.~Khuller.
\newblock Approximation algorithms for finding highly connected subgraphs.
\newblock In D.~Hochbaum, editor, {\em Approximation Algorithms for {NP}-hard
  problems}, chapter~6, pages 236--265. PWS, 1995.

\bibitem{KCJ}
Z.~Kir\'{a}ly, B.~Cosh, and B.~Jackson.
\newblock Local edge-connectivity augmentation in hypergraphs is {NP}-complete.
\newblock {\em Discrete Applied Mathematics}, 158(6):723--727, 2010.

\bibitem{KN-TAP}
G.~Kortsarz and Z.~Nutov.
\newblock A simplified $1.5$-approximation algorithm for augmenting
  edge-connectivity of a graph from 1 to 2.
\newblock {\em ACM Transactions on Algorithms}, 12(2):23, 2016.

\bibitem{MN}
Y.~Maduel and Z.~Nutov.
\newblock Covering a laminar family by leaf to leaf links.
\newblock {\em Discrete Applied Mathematics}, 158(13):1424--1432, 2010.

\bibitem{N-Th}
Z.~Nutov.
\newblock {\em Structures of Cuts and Cycles in Graphs; Algorithms and
  Applications}.
\newblock PhD thesis, Technion, Israel Institute of Technology, 1997.

\bibitem{N}
Z.~Nutov.
\newblock Approximating connectivity augmentation problems.
\newblock {\em ACM Trans. Algorithms}, 6(1):5:1--5:19, 2009.

\bibitem{N-T}
Z.~Nutov.
\newblock On the tree augmentation problem.
\newblock In {\em ESA}, pages 61:1--61:14, 2017.
\newblock To appear in Algorithmica.

\bibitem{N-2c}
Z.~Nutov.
\newblock $2$-node-connectivity network design.
\newblock {\em CoRR}, abs/2002.04048, 2020.
\newblock To appear in WAOA20.
\newblock URL: \url{https://arxiv.org/abs/2002.04048}.

\bibitem{NS}
Z.~Nutov, G.~Kortsarz, and E.~Shalom.
\newblock Approximating activation edge-cover and facility location problems.
\newblock In {\em MFCS}, pages 20:1--20:14, 2019.

\bibitem{RW}
R.~Ravi and D.~P. Williamson.
\newblock An approximation algorithm for minimum-cost vertex-connectivity
  problems.
\newblock {\em Algorithmica}, 18(1):21--43, 1997.

\end{thebibliography}

\end{document}